\documentclass[runningheads]{llncs}
\usepackage{amsmath}
\usepackage{amssymb}
\usepackage{mathrsfs}
\usepackage{bm}
\usepackage{tikz}
\usepackage{multirow}
\usetikzlibrary{calc}
\let\emptyset=\varnothing

\newcommand{\powerset}[1]{2^{#1}}

\newcommand{\Arena}{(P,V,\allowbreak (V_p)_{p\in P},\allowbreak v_0,\allowbreak \Delta)}
\newcommand{\Play}{\mathit{Play}}
\newcommand{\Hist}{\mathit{Hist}}
\newcommand{\Pos}{\mathsf{Pos}}

\newcommand{\Strategies}{\mathsf{Str}}
\newcommand{\out}{\mathrm{out}}
\newcommand{\Inf}{\mathit{Inf}}
\newcommand{\Buchi}{\mathop{\mathrm{B\ddot{u}chi}}\nolimits}

\newcommand{\SVER}{\mathit{sVP}}

\newcommand{\SAT}{\mathsf{SAT}}
\newcommand{\VPCKR}{\mathit{VPCKR}}
\newcommand{\calK}{\mathop{\mathcal{K}}\nolimits}
\newcommand{\calMK}{\mathop{\mathcal{MK}}\nolimits}
\newcommand{\calCK}{\mathop{\mathcal{CK}}\nolimits}
\newcommand{\RAT}{\mathit{RAT}}

\newcommand{\Str}{\mathsf{Str}}

\newcommand{\Tpl}[1]{\langle{#1}\rangle}
\newcommand{\Deriv}[3]{{#1}[#2\mapsto #3]}

\newcommand{\classP}{\mathsf{P}}
\newcommand{\PTIME}{\mathrm{PTIME}}
\newcommand{\NP}{\mathrm{NP}}
\newcommand{\coNP}{\text{co}\NP}
\newcommand{\coNEXP}{\text{co}\mathrm{NEXP}}
\newcommand{\classPi}{\mathrm{\Pi}^{\classP}}
\newcommand{\classSigma}{\mathrm{\Sigma}^{\classP}}
\newcommand{\balpha}{{\boldsymbol\alpha}}
\newcommand{\bsigma}{{\boldsymbol\sigma}}
\newcommand{\bolds}{{\boldsymbol s}}
\newcommand{\boldt}{{\boldsymbol t}}
\newcommand{\VPNash}{\mathit{VPNash}}

\newcommand{\aeSAT}{{\forall\exists\mathsf{SAT}}}
\newcommand{\eaSAT}{{\exists\forall\mathsf{SAT}}}
\newcommand{\Out}{\mathop{\mathrm{out}}\nolimits}
\newcommand{\Nash}{\mathop{\mathit{Nash}}\nolimits}
\newcommand{\Win}{\mathop{\mathrm{Win}}\nolimits}
\newcommand{\Muller}{\mathop{\mathrm{Muller}}\nolimits}
\newcommand{\True}{\mathit{true}}
\newcommand{\False}{\mathit{false}}

\let\OL=\overline
\let\emptyset=\varnothing

\title{Verification with Common Knowledge of Rationality for Graph Games}
\author{
	Rindo Nakanishi\inst{1} \and
	Yoshiaki Takata\inst{2} \and
	Hiroyuki Seki\inst{1}
}
\authorrunning{R. Nakanishi et al.}
\institute{
	Graduate School of Informatics, Nagoya University \\
	Furo-cho, Chikusa, Nagoya 464-8601, Japan \\
	\email{\{rindo,seki\}@sqlab.jp}
	\and School of Informatics, Kochi University of Technology \\
	Tosayamada, Kami City, Kochi 782-8502, Japan \\
	\email{takata.yoshiaki@kochi-tech.ac.jp}
}

\begin{document}
\maketitle
	\begin{abstract}
		Realizability asks whether there exists a program satisfying its specification.
In this problem, we assume that each agent has her own objective and behaves rationally to satisfy her objective.
Traditionally, the rationality of agents is modeled by a Nash equilibrium (NE),
where each agent has no incentive to change her strategy
because she cannot satisfy her objective by changing her strategy alone.
However, an NE is not always an appropriate notion for the rationality of agents
because the condition of an NE is too strong;
each agent is assumed to know strategies of the other agents completely.
In this paper, we use an epistemic model to define 
common knowledge of rationality  of all agents (CKR).
We define the verification problem as a variant of the realizability problem,
based on CKR, instead of NE.
We then analyze the complexity of the verification problems for the class of positional strategies.
	\end{abstract}
\section{Introduction}
	A graph game is a formal model for analyzing or controlling  
a system consisting of multiple agents (or processes) that 
behave independently according to their own preferences or objectives.
One of the useful applications of graph game is {\em reactive synthesis},
which is the problem of synthesizing a reactive system 
that satisfies a given specification. 
The standard approach to the problem is as follows \cite{BCJ18}. 
When a specification is given by a linear temporal logic (LTL) formula 
(or nondeterministic $\omega$-automaton) 
$\varphi$, we translate $\varphi$ to 
an equivalent deterministic $\omega$-automaton ${\cal A}$. 
Next, we convert ${\cal A}$ to a tree automaton (or equivalently, a parity game) ${\cal B}$.
Then, we test whether the language recognized by $\mathcal{B}$ is empty, i.e., $L(\mathcal{B})=\emptyset$
(or equivalently, there is a winning strategy for player $0$, which is the system player in $\mathcal{B}$).
The answer to the problem is affirmative if and only if $L({\cal B})\not=\emptyset$, 
and any $t\in L({\cal B})$ (or any winning strategy for the system player in ${\cal B}$)
is an implementation of the specification. 

As described above, reactive synthesis can be viewed as a two-player zero-sum game, in which 
the system player aims at satisfying the specification as her goal (or winning objective)
whereas the objective of the environment player is the negation of the specification. 
If there is a winning strategy for the system player, then any one of them is 
an implementation satisfying the specification. 
However, the assumption that the objective of the environment is antagonistic to 
the system's objective is too conservative;  
usually, the environment behaves based on its own preference or interest. 
Also, the environment often consists of multiple agents, and hence 
a multi-player non-zero-sum game is a more appropriate model than a two-player zero-sum game.
Furthermore, it is natural to require that the system should satisfy the specification 
under the assumption that all players behave {\em rationally}, i.e., 
they behave aiming to satisfy their own objectives. 

{\em Rational synthesis} (abbreviated as RS) asks whether
a given specification is satisfied whenever all the players behave rationally. 
{\em Rational verification} (abbreviated as RV) is the problem asking whether 
every rational strategy profile satisfies a specification. 
RV is defined by adding the rationality assumption on the usual model checking, 
which asks whether every execution of a given model satisfies a specification \cite{CGP01}. 
Note that RV and RS are closely related.  
The answer of RV for a specification $\psi$ is {\em no} if and only if 
the answer of RS for $\neg \psi$ is {\em yes}.  
Namely, RV fails for $\psi$
iff there is a counter-example to $\psi$ (a rational behavior that satisfies $\neg\psi$). 
As described in the related work section below, 
rationality is traditionally captured by Nash equilibrium, 
which is one of the most important concepts in game theory.
We say that a tuple of strategies of all players (called a strategy profile) 
is a {\em Nash equilibrium} (abbreviated as NE)
when no one can improve her own payoff, which is the reward she receives from the game,
by changing her strategy alone.
An NE locally maximizes each player's payoff and hence
each player has no incentive to change her strategy.
From the viewpoint of epistemic game theory, however, 
NE is not always suitable for the concept of rationality
because each player is assumed to know the strategies of the other players.  
(Also see the related work below.)

Epistemic game theory~\cite{Bo15} uses a Kripke frame that consists of 
a set of worlds (or states) $W$ and 
a subset $R_p(w)\subseteq W$ for each player $p$ and a world $w$. 
For a world $w$ and a player $p$, 
$R_p(w)$ represents the set of possible worlds from the viewpoint of $p$ in 
the actual world $w$ of $p$.  
For instance, if $R_p(w) = \{w, w', w''\}$, 
it means that the information given to $p$ in $w$ is incomplete and 
$p$ cannot distinguish $w$ from $w'$ or $w''$.  
Possible world is useful for modelling a situation such that 
each player (or process) cannot know the internal states of the other players
(e.g., the contents of local variables of the other processes). 
An {\em epistemic model} is a pair of a Kripke frame and 
a mapping that associates a strategy profile with each world. 
We say that a player $p$ is rational in a world $w$, 
if for every possible world $w'$ of $p$ in $w$, 
there is no better strategy of $p$ than the one associated with $w$. 
Then, a strategy profile is said to be {\em epistemically rational} if 
there exist an epistemic model $M$ and a world $w$ in $M$ such that 
every player is rational in $w$ and this property is a common knowledge 
among all players. 
(The formal definition of common knowledge is postponed to the next section.)

In this paper, 
we propose a new framework for reactive synthesis and verification, 
by augmenting graph game with epistemic models. 
We then define the rational verification problem based on the proposed model
and present some results on the complexity of the problem.
%
\subsection*{Related work} 
Studies on reactive synthesis has its origin in 1960s and has been one of central topics
in formal methods as well as model checking. 
The problem is EXPTIME-complete
when a specification is given by an $\omega$-automaton \cite{BL69}
and 2EXPTIME-complete 
when a specification is given by an LTL formula \cite{PR89}. 
Rational synthesis (RS) is in 2EXPTIME \cite{FKL10}
when a specification and objectives (in what follows, we refer to as objectives only)
are given as LTL formulas. 
It is PSPACE-complete when objectives of players are restricted to GR(1) \cite{GNPW23}.
The complexity of RS is also studied for $\omega$-regular objectives. 
RS is PTIME-complete with B\"{u}chi objectives, 
NP-complete with coB\"{u}chi, parity, Streett objectives \cite{Um08} 
and PSPACE-complete with Muller objectives \cite{CFGR16}. 
RS has been applied to the synthesis of non-repudiation and fair exchange protocols \cite{KR01,CR12}. 
RS is optimistic in the sense that 
the system first proposes a strategy profile and all environment players will
follow it as far as they do not have profitable deviations. 
For this reason, another type of RS was proposed, called 
non-cooperative rational synthesis (NCRS) in \cite{KPV16}.
NCRS asks whether there is a strategy $s_0$ of the system such that 
every 0-fixed NE (a strategy profile where no environment player 
has a profitable deviation) including $s_0$ satisfies the specification. 
Decidability and complexity of NCRS have also been studied \cite{CFGR16,KPV16,KS22}.

Bruy\`{e}re, et al.\,\cite{BRT22} investigated the complexity of 
rational verification (RV) 
taking Pareto-optimality as the notion of rationality, and show that 
RV is coNP-complete, $\classPi_2$-complete and PSPACE-complete with 
parity, Boolean B\"{u}chi and LTL objectives, respectively. 
Brice, et al.\,\cite{BRB23} considered weighted (or duration) games, 
adopted NE and subgame-perfect equilibrium as the notions of rationality 
and showed that RV is coNP-complete with mean payoff objectives and 
undecidable with energy objectives. 

Epistemic rationality does not always imply NE.
In~\cite{AB95}, Aumann and Brandenburger show an epistemic sufficient condition for NE
in terms of strategic form game (not graph game).
(Also see \cite{P99}.)

As described above, there are already many studies on the decidability and complexity 
of RS and RV. 
However, all of them take NE or its refinement as criteria of rationality. 
This paper is the first step for defining and analyzing RV where 
the rationality is defined in an epistemic way. 
Also, we think that combining an epistemic model with a usual graph game 
enables us to express 
incomplete information of a player in a natural way.

	\subsection*{Outline}
		In Sections~\ref{sec:graph_game} and~\ref{sec:epistemic_model},
we review graph game and epistemic model.
In Section~\ref{sec:VPCKR}, we define
rational verification problems $\VPCKR_S$, $\VPCKR_{\classP,S}$ and $\VPNash_S$.
The problem $\VPCKR_S$ asks whether all strategy profiles over the class $S$ of strategies
satisfy a given specification when CKR holds.
The problem $\VPCKR_{\classP,S}$ is a variant of $\VPCKR_S$, where
the size of an epistemic model is not greater than a polynomial size of a game arena.
The problem $\VPNash_S$ asks whether all NE over the class $S$ of strategies
satisfy a given specification.
Table~\ref{tab:complexity} shows the complexities of these problems.
\begin{table}[t]
    \caption{The complexities of verification problems}
    \centering
    \begin{tabular}{|c|c|c|c|c|}
        \hline
        & $\VPCKR_\Pos$ & $\VPCKR_{\classP,\Pos}$ & $\VPNash_\Pos$ &  $\VPNash_\Strategies$\\
        \hline
        Upper bound & $\coNEXP^{\NP}$ & $\classPi_2$ & \multirow{2}{*}{$\classPi_2$-complete}
            & \multirow{2}{*}{$\mathrm{PSPACE}$-complete~\cite{CFGR16}} \\
        \cline{1-3}
        Lower bound & $\classSigma_2$-hard & $\coNP$-hard &&\\ 
        \hline
    \end{tabular}
    \label{tab:complexity}
\end{table}
$\Str$ is the class of all strategies and $\Pos$ is the class of all positional strategies.
In Section~\ref{sec:conclusion},
we summarize the paper and give future work.
\section{Graph Game}\label{sec:graph_game}
	In this section, we provide basic definitions and notions on graph games, which are needed 
to present our new framework.  
We start with the definition of game arena, winning objective and strategy and so on, 
followed by the definition of Nash equilibrium (NE). 

A graph game is a directed graph with an initial vertex.
Each vertex is controlled by a player who chooses next vertex.
A game is started from the initial vertex.
Then, players repeatedly choose the next vertex according to their strategies.
An infinite sequence of vertices generated by such a process is called a play.
If the play satisfies the winning objective of a player, then she wins.
Otherwise, she loses. 
Note that our setting is non-zero-sum, 
hence it is possible that there are multiple winners.
An NE is a tuple of strategies of all players where each loser cannot become a winner
by changing her strategy alone.

For a binary relation $R \subseteq X \times X$ over a set $X$ and a subset $A\subseteq X$,
we define $R(A) = \{ x \in X \mid (a,x) \in R \wedge a \in A \} \subseteq X$.
	\paragraph{Game arena}
		\begin{definition}
    A game arena is a tuple $G = \Arena$, where
    \begin{itemize}
        \item $P$ is a finite set of players,
        \item $V$ is a finite set of vertices,
        \item $(V_p)_{p\in P}$ is a partition of $V$,
              namely, $V_i \cap V_j = \emptyset$ for all $i\neq j \ (i,j \in P)$ and
              $\bigcup_{p\in P} V_p = V$,
        \item $v_0 \in V$ is the initial vertex, and
        \item $\Delta \subseteq V \times V$ is a set of edges
              where $\Delta(v) \neq \varnothing$ for all $v\in V$.
    \end{itemize}
\end{definition}
	\paragraph{Play and history}
		An infinite sequence of vertices $v_0 v_1 v_2 \cdots \ (v_i \in V, i \geq 0)$
starting from the initial vertex $v_0$ is a \emph{play}
if $(v_i,v_{i+1}) \in \Delta$ for all $i \geq 0$.
A \emph{history} is a non-empty (finite) prefix of a play.
The set of all plays is denoted by $\Play_G$ and the set of all histories is denoted by $\Hist_G$.
We often write a history as $hv$ where
$h \in \Hist \cup \{\varepsilon\}$ and $v \in V$.
For a player $p \in P$, let $\Hist_{G,p} = \{ hv \in \Hist \mid v \in V_p \}$.
That is, $\Hist_{G,p}$ is the set of histories ending with a vertex controlled by player $p$.
We abbreviate $\Play_G$, $\Hist_{G,p}$ and $\Hist_G$
as $\Play$, $\Hist_p$ and $\Hist$ respectively, if $G$ is clear from the context.
For a play $\rho=v_0v_1v_2\cdots \in \Play$,
we define $\Inf(\rho) = \{ v\in V\mid \forall i\geq0.\ \exists j\geq i.\ v_j=v\}$.
	\paragraph{Strategy}
		For a player $p \in P$, a \emph{strategy} of $p$ is a function $s_p: \Hist_p \to V$
such that $(v, s_p(hv)) \in \Delta$ for all $hv \in \Hist_p$.
At a vertex $v\in V_p$,
player $p$ chooses $s_p(hv)$ as the next vertex according to her strategy $s_p$. 
Note that because the domain of $s_p$ is $\Hist_p$,
the next vertex may depend on the whole history in general.
Let $\Strategies_{G,p}$ denote the set of all strategies of $p$.

When a $p$'s strategy $s_p \in \Strategies_{G,p}$ satisfies
$s_p(hv) = s_p(h'v)$ for all $hv,h'v \in \Hist_p$,
we say that $s_p$ is \emph{positional} because the next vertex depends only on the current vertex $v$.
We regard a function $s_p: V_p \to \Delta(V)$
as a $p$'s positional strategy where $s_p(hv) = s_p(v)$ for all $hv \in \Hist_p$.
Let $\Pos_{G,p} \subseteq \Strategies_{G,p}$ denote the set of all positional strategies of $p$.


We abbreviate $\Strategies_{G,p}$ and $\Pos_{G,p}$ as
$\Strategies_p$ and $\Pos_p$ respectively, if $G$ is clear from the context.
	\paragraph{Strategy profile}
		A \emph{strategy profile} is a tuple $\bm{s} = (s_p)_{p \in P}$
of strategies of all players, namely $s_p \in \Strategies_p$ for all $p \in P$.
Let $\Strategies_G$ (resp. $\Pos_G$) be the set of all strategy profiles
(resp. the set of all strategy profiles ranging over positional strategies).
We define the function $\out_G:\Strategies_G \to \Play$ as $\out_G((s_p)_{p \in P}) = v_0 v_1 v_2 \cdots$
where $v_{i+1} = s_p(v_0\cdots v_i)$ for all $i \geq 0$ and for $p\in P$ with $v_i \in V_p$.
We call the play $\out_G(\bm{s})$ the \emph{outcome} of~$\bm{s}$.
We abbreviate $\Strategies_G$, $\Pos_G$ and $\out_G$ as
$\Strategies$, $\Pos$ and $\out$ respectively, if $G$ is clear from the context.
For a strategy profile $\bm{s} \in \Strategies_G$ and
a strategy $s'_p \in \Strategies_p$ of a player $p\in P$,
let $\bm{s}[p \mapsto s'_p] \in \Strategies_G$ denote the strategy profile obtained from $\bm{s}$ by replacing 
the strategy of $p$ in $\bm{s}$ with $s'_p$.
	\paragraph{Objective}
		We assume that the result a player obtains from a play is either a winning or a losing.
Each player has her own winning condition over plays, 
and we represent a winning condition by a subset $O\subseteq\Play$ of plays; 
i.e., the player wins if and only if the play belongs to the subset $O$.
We call the subset $O$ the \emph{objective} of that player.
In this paper, we focus on the following important classes of objectives.
\begin{definition}\label{def:obj}
Let $U\subseteq V$ be a subset of vertices and
$\varphi$ be a Boolean formula whose variables are the vertices of $V$. 
We will use $U$ and $\varphi$ as finite representations for specifying an objective as follows.
\begin{itemize}
	\item B\"{u}chi objective: \\$\Buchi(U)=\{ \rho\in\Play \mid \Inf(\rho) \cap U \neq \varnothing\}$.
	\item Muller objective:
	\\$\Muller(\varphi)=\{ \rho \in \Play \mid
		\varphi\text{ is }\True \text{ under }\theta_\rho\}$ where
		$\theta_\rho$ is the truth assignment defined as $\theta_\rho(v)=\True$ iff $v\in\Inf(\rho)$.
\end{itemize}
\end{definition}
Note that a $\Buchi$ objective is also a Muller objective:
For any $U\subseteq V$, it holds that $\Buchi(U)=\Muller(\bigvee_{u\in U}u)$.

	\paragraph{Objective profile}
		An \emph{objective profile} is a tuple $\bm{\alpha} = (O_p)_{p \in P}$ of objectives of all players,
namely $O_p \subseteq \Play$ for all $p \in P$.
For a strategy profile $\bm{s} \in \Strategies$ and an objective profile $\bm{\alpha} = (O_p)_{p \in P}$,
we define the set $\Win_G(\bm{\alpha},\bm{s}) \subseteq P$ of winners as
$\Win_G(\bm{\alpha},\bm{s}) = \{ p \in P \mid \out_G(\bm{s}) \in O_p \}$.
That is, a player $p$ is a winner if and only if  $\out_G(\bm{s})$ belongs to the objective $O_p$ of $p$. 
If $p \in \Win_G(\bm{\alpha},\bm{s})$, we also say that $p$ wins for $G$ and $\bm{\alpha}$
(by the strategy profile $\bm{s}$). 
Note that it is possible that no player wins the game or all the players win the game. 
In this sense, a game is \emph{non-zero-sum}. 
If an objective profile $\bm{\alpha} = (O_p)_{p \in P}$ is a partition of $\Play$, i.e.,
$O_i \cap O_j = \varnothing$ for all $i\neq j \ (i,j\in P)$ and $\bigcup_{p\in P} O_p = \Play$,
then the game is called \emph{zero-sum}.
When a game is zero-sum, there is one and only one winner and the other players are all losers.
We abbreviate $\Win_G$ as $\Win$ if $G$ is clear from the context.
	\paragraph{Winning strategy}
		Let $S \in \{ \Pos,\Strategies \}$ be a class of strategy profiles.
When the objective of player~$p$ is~$O_p$,
a strategy $s\in S_p$ is called
a \emph{winning strategy} of~$p$ if
it holds that
$\Out(\Deriv{\bolds}{p}{s}) \in O_p$
for every strategy profile $\bolds\in S$.
That is,
$s$~is a winning strategy of~$p$ if
$p$ always wins by taking $s$
regardless of the strategies of the other players.
	\paragraph{Nash equilibrium}
		Let $\bm{\alpha}=(O_p)_{p\in P}$ be an objective profile and
$S \in \{ \Pos,\Strategies \}$ be a class of strategy profiles.
A strategy profile $\bm{s} \in S$ is called a \emph{Nash equilibrium} (NE) for $\bm{\alpha}$ and $S$
if it holds that
$\forall p \in P.\ \forall s_p \in S_p.\ p\in\Win(\bm{\alpha},\bm{s}[p\mapsto s_p]) 
\implies p \in\Win(\bm{\alpha},\bm{s})$.
Intuitively, $\bm{s}$ is an NE if any player $p$ cannot improve the result (from losing to winning) 
by changing her strategy alone.
Because $p\in \Win(\bm{\alpha},\bm{s})$ is equivalent to $\out(\bm{s}) \in O_p$, 
a strategy profile $\bm{s} \in S$ is an NE for $\bm{\alpha}$ and $S$ if and only if 
$\forall p\in P.\ \forall s_p\in S_p.\ \out(\bm{s}[p\mapsto s_p])\in O_p
\implies \out(\bm{s})\in O_p$.
We write this condition as $\Nash(\bm{s},\bm{\alpha},S)$.
		\begin{example}\label{exm:graph_game_example}
    Figure~\ref{fig:graph_game} shows a $3$-player game arena
    $G=\Arena$ where
    $P = \{ 0, 1, 2\}$,
    $V = \{ v_0,v_1,v_2\}$,
    $V_p = \{ v_p \} \ (p\in P)$ and 
    $\Delta = \{ (v_i,v_j) \mid i,j \in P,\ i\neq j \}$.
    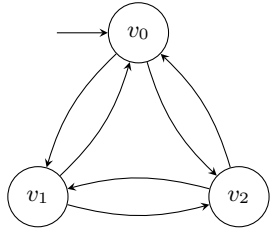
\begin{figure}[t]
		\centering
		\begin{tikzpicture}[everynode/.style={circle,draw,minimum size=0.8cm},>=stealth]
            \node[everynode] (v0) at ({90}: 1.5) {$v_0$};
            \node[everynode] (v1) at ({207}: 1.5) {$v_1$};
            \node[everynode] (v2) at ({333}: 1.5) {$v_2$};
            \node[left of=v0, node distance=1.2cm] (start) {};
            \draw[->] (start) to (v0);
            \draw[->] (v0) to [bend right=15] (v1);
            \draw[->] (v1) to [bend right=15] (v2);
            \draw[->] (v2) to [bend right=15] (v0);
            \draw[->] (v0) to [bend right=15] (v2);
            \draw[->] (v2) to [bend right=15] (v1);
            \draw[->] (v1) to [bend right=15] (v0);
		\end{tikzpicture}
			\caption{$3$-player game arena with B\"{u}chi objectives}
			\label{fig:graph_game}
	\end{figure}
    The objective of player $p$ is $O_p = \Buchi(\{v_{(p+1)\bmod 3}\})$,
    namely to visit the vertex $v_{(p+1)\bmod 3}$ infinitely often.
    The objective profile is $\bm{\alpha} = (O_p)_{p\in P}$.
    Let $\bm{s}=(s_p)_{p\in P} \in \Pos$
    be the strategy profile over positional strategies
    where $s_p(h) = v_{(p+1)\bmod 3}$ for all $h \in \Hist_p$.
    Let $\bm{s'}=(s'_p)_{p\in P} \in \Pos$
    be the strategy profile over positional strategies
    where $s'_0(h_0)=v_1$ and $s'_1(h_1) = s'_2(h_2) =v_0$
    for all $h_p \in \Hist_p \ (p\in P)$.
    It holds that $\out(\bm{s}) = (v_0 v_1 v_2)^\omega \in O_p$ for all $p$.
    Hence, $\Win(\bm{\alpha},\bm{s}) = \{ 0,1,2 \}$.
    On the other hand, it holds that $\out(\bm{s'}) = (v_0 v_1)^\omega \in O_0 \cap O_2$ and
    $\out(\bm{s'}) \notin O_1$.
    Hence, $\Win(\bm{\alpha},\bm{s'}) = \{ 0,2 \}$.
    The strategy profile $\bm{s}$ is an NE for $\bm{\alpha}$.
    The strategy profile $\bm{s'}$ is not an NE for $\bm{\alpha}$
    because there is a positional strategy $s_1 \in \Pos_1$ of player $1$ such that
    $1 \in \Win(\bm{\alpha},\bm{s'}[1\mapsto s_1])$.
\end{example}
	\paragraph{}
			~~~~~The following problem asks
if a given strategy profile $\bm{s}$ satisfies a specification $O$.
We use Lemma~\ref{lem:sver_pos} to prove the upper bounds of the complexities of verification problems
in Section~\ref{sec:VPCKR}.
\begin{problem}\label{prob:vps}
    Let $G$ be a game arena,
    $O \subseteq \Play$ be a Muller objective and
    $S\in \{ \Pos,\Strategies \}$ be a class of strategy profiles.
    We define the simple verification problem as follows.
    \[
        \SVER_S = \{ \langle G,\bm{s},O \rangle \mid \bm{s} \in S \wedge \out(\bm{s}) \in O \}.
    \]
\end{problem}

\begin{lemma}\label{lem:sver_pos}
    $\SVER_\Pos$ is in $\PTIME$.
\end{lemma}


\begin{proof}
    Let $\varphi$ be a given Boolean formula
    representing a Muller objective~$O$
    (i.e.\ $O=\Muller(\varphi)$).
    Because $\bm{s} \in \Pos$ is a strategy profile over positional strategies,
    the play $\out(\bm{s})$ can be written as $\out(\bm{s}) = u_0 u^\omega$
    for some $u_0 \in V^*$ and $u \in V^+$
    such that $u_0 u$ does not contain any vertex twice.
    Vertices in $u$ are visited infinitely often
    and vertices not in $u$ are visited only finite times,
    and thus we can construct the truth assignment
    $\theta_{\out(\bm{s})}$ such that
    $\theta_{\out(\bm{s})}(v)=\True$ iff
    $v\in\Inf(\out(\bm{s}))$ in polynomial time.
    Simply evaluating $\varphi$ under $\theta_{\out(\bm{s})}$,
    we can check whether
    the play satisfies the Muller objective~$O$.
\qed
\end{proof}
\section{Epistemic Model}\label{sec:epistemic_model}
	In this section, we first review Kripke frame and epistemic model
together with the important notion: 
(epistemic) rationality and common knowledge of rationality
and give simple examples. 
We then propose a new characterization of the notion of common knowledge of 
rationality based on graph games. 
	\paragraph{KT5 Kripke frame}
		\begin{definition}
    A KT5 Kripke frame is a pair $(W,(R_p)_{p\in P})$, where
    \begin{itemize}
        \item $P$ is a finite set of players,
        \item $W$ is a finite set of (possible) worlds, and
        \item $R_p \subseteq W \times W$ is an equivalence relation on $W$, namely, $R_p$ satisfies
              \begin{description}
                \item[(reflexivity)] $\forall w \in W.\ (w,w) \in R_p$,
                \item[(symmetry)] $\forall w,w'\in W.\ \left((w,w') \in R_p \implies (w',w) \in R_p \right)$, and
                \item[(transitivity)] $\forall w_1,w_2,w_3 \in W.\ \left((w_1,w_2),(w_2,w_3) \in R_p \implies (w_1,w_3) \in R_p \right)$.
              \end{description}
    \end{itemize}
\end{definition}
A Kripke frame expresses the structure of knowledge of players.
In the world $w$, Player $p$ only knows that she is in one of the worlds of $R_p(w)$.
In other words,
in the world $w$, Player $p$ cannot distinguish the worlds of $R_p(w)$ with one another.
	\paragraph{Knowledge operator}
		For a given KT5 Kripke frame $(W,(R_p)_{p\in P})$,
we call any subset $E \subseteq W$ an \emph{event}.
\begin{definition}
    Let $(W,(R_p)_{p\in P})$ be a KT5 Kripke frame and
    $p \in P$ be a player.
    The \emph{knowledge operator} $\calK_p:\powerset{W}\to \powerset{W}$,
    the \emph{mutual knowledge operator} $\calMK:\powerset{W}\to \powerset{W}$ and 
    the \emph{common knowledge operator} $\calCK:\powerset{W}\to \powerset{W}$
    are defined as follows.
    \begin{align*}
        \calK_p(E) &= \{w \in W \mid R_p(w) \subseteq E \},\\
        \calMK(E) &= \bigcap_{p\in P} \calK_p(E), \ \text{and}\\
        \calCK(E) &= \bigcap_{1\leq i} \calMK^i(E),
    \end{align*}
    where $\calMK^i\ (0 \leq i)$ is defined as
    \begin{align*}
        \calMK^0(E) &= E, \ \text{and} \\
        \calMK^{i+1}E &= \calMK(\calMK^i(E)).
    \end{align*}
\end{definition}
Equivalently, we can define $\calCK$ as $\calCK(E) = \{ w \in W \mid R^+(w)\subseteq E\}$
where $R^+$ is the transitive closure of $\bigcup_{p\in P}R_p$.
Note that there is no constant upper bound on the depth of the recursive definition of $\calCK$.

Recall that in a world $w$, she knows only that she is in one of the worlds of $R_p(w)$.
If $w \in \calK_p(E)$, then $R_p(w) \subseteq E$ holds from the definition of $\calK_p$.
Hence, in a world $w \in \calK_p(E)$, player $p$ knows that she is in one of the worlds of $E$.
When $w\in \calK_p(E)$, we say that player $p$ knows the event $E$ occurs in $w$,
or simply $p$ knows $E$ in $w$.
The set $\calMK(E)$ is the event such that all players know the event $E$.
If $w \in \calMK(E)$, we say that the event $E$ is mutual knowledge in $w$.

If all players know an event $E$ and all players know that all players know the event $E$,
and all players know that all players know that all players know the event $E$ and so on,
we say that $E$ is common knowledge.
The set $\calCK(E)$ is the event where $E$ is common knowledge.
If $w \in \calCK(E)$, we say that the event $E$ is common knowledge in $w$.
	\paragraph{Epistemic model}
		\begin{definition}
    Let $G=\Arena$ be a game arena and $S\in \{ \Pos,\Strategies \}$ be a class of strategy profiles.
    An \emph{epistemic model} for $G$ and $S$ is a tuple $(W,(R_p)_{p\in P},(\sigma_p)_{p\in P})$
    where $(W,(R_p)_{p\in P})$ is a KT5 Kripke frame and
    $\sigma_p:W\to S$ is a function such that
    \begin{equation}\label{eqn:epistemic_model}
        \forall w,w'\in W.\ \left( (w,w')\in R_p \implies \sigma_p(w) = \sigma_p(w')\right).
    \end{equation}
\end{definition}
Condition~(\ref{eqn:epistemic_model}) guarantees that
each player takes the same strategy in the worlds she cannot distinguish.
For a game arena $G$ and a class of strategies $S\in \{ \Pos,\Strategies \}$,
let $M(G,S)$ be the set of all epistemic models for $G$ and $S$.
	\paragraph{}
		\begin{example}\label{exm:model}
    Let $G$ be a game arena defined in Example~\ref{exm:graph_game_example}.
    Let $s^R_p, s^L_p \in \Pos_p$ be the positional strategies for player $p$ defined as
    $s^R_p(h) = v_{(p-1)\bmod 3}$ and
    $s^L_p(h) = v_{(p+1)\bmod 3}$
    for all $h\in \Hist_p$.
    Let $M=(W,(R_p)_{p\in P},\bm{\sigma}) \in M(G,\Pos)$ be an epistemic model
    for $G$ and $\Pos$,
    where
    $W = \{ RRR, RRL, RLR$, $RLL, LRR$, $LRL, LLR, LLL \}$,
    $R_p = \{ (X_0 X_1 X_2, Y_0 Y_1 Y_2) \mid$
    $X_i, Y_i \in \{R, L\}$ for $0 \leq i \leq 2$ and $X_p = Y_p\}$,
    $\bm{\sigma}(XYZ) = (s^X_0,s^Y_1,s^Z_2)$ for all $X,Y,Z\in \{ R,L \}$.
      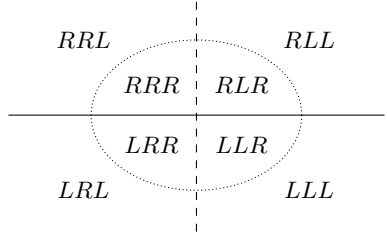
\begin{figure}[t]
        \centering
        \begin{tikzpicture}
          \foreach \i/\vi [evaluate=\i as \x using 0.6*\i] in {-1/R, 1/L}
          {
            \foreach \j/\vj [evaluate=\j as \y  using 0.4*\j,
                             evaluate=\x as \xx using 2.5*\x,
                             evaluate=\y as \yy using 2.5*\y] in {1/R, -1/L}
            {
              \node at (\x,\y)   {$\vj\vi{}R$};
              \node at (\xx,\yy) {$\vj\vi{}L$};
            }
          }
          \draw (-2.5,0) -- (2.5,0);
          \draw[dashed] (0,-1.55) -- (0,1.55);
          \draw[densely dotted] (0,0) ellipse (1.4 and 1);
        \end{tikzpicture}
        \caption{Equivalence classes in $W$}
            \label{fig:model_rat_example}
      \end{figure}
      Figure~\ref{fig:model_rat_example} shows the equivalence classes in $W$.
      The solid, dashed and densely dotted lines divide $W$ into
      the equivalence classes of $R_0$, $R_1$ and $R_2$ respectively.
      Let $E=\{ RRR, RRL, RLR, RLL \}$ be the event such that
      player $0$ takes the strategy $s^R_0$.
      Then, player $0$ knows $E$ in each world of $E$ because $\calK_0(E) = E$.
      On the other hand, players $1$ and $2$ never know $E$ in any world of $W$ 
      because $\calK_1(E) = \calK_2(E) = \varnothing$.
      The event $W$ is both mutual knowledge and common knowledge
      in all worlds because it holds that $\calMK(W)=\calCK(W) = W$.
\end{example}
	\paragraph{Rationality}
		\begin{definition}\label{def:rat}
    Let $G=\Arena$ be a game arena,
    $\bm{\alpha}$ be an objective profile,
    $S \in \{ \Pos, \Strategies \}$ be a class of strategy profiles and
    $M = (W,(R_p)_{p\in P}, \bm{\sigma})$ be an epistemic model for $G$ and $S$.
    For a world $w\in W$ and a player $p \in P$,
    if there is no $p$'s strategy $s_p \in S$ such that
    \begin{align}
        &\forall w'\in R_p(w).\
            \left(
                p \in \Win(\bm{\alpha},\bm{\sigma}(w'))
                \implies p \in \Win(\bm{\alpha},\bm{\sigma}(w')[p\mapsto s_p])
            \right), \text{and} \label{align:rat1}\\
        &\exists w'\in R_p(w).\
            \left(
                p \notin \Win(\bm{\alpha},\bm{\sigma}(w')) \wedge
                p \in \Win(\bm{\alpha},\bm{\sigma}(w')[p\mapsto s_p])
            \right),\label{align:rat2}
    \end{align}
    then $p$ is \emph{rational}\footnote{%
        The rationality in Definition~\ref{def:rat} is called the strong notion of rationality~\cite{Bo15}.
    } in $w$.
\end{definition}
We write the set of all worlds where $p$ is rational as $\RAT^p_{G,\bm{\alpha},M,S}\subseteq W$.
The set of all worlds where each player is rational is written as
$\RAT_{G,\bm{\alpha},M,S} = \bigcap_{p\in P}\RAT^p_{G,\bm{\alpha},M,S}$.
	\paragraph{Characterization of the notion of common knowledge of rationality}
		\begin{definition}
    Let $G=\Arena$ be a game arena,
    $\bm{\alpha}$ be an objective profile,
    $S \in \{ \Pos, \Strategies \}$ be a class of strategy profiles.
    We define a \emph{characterization $T_{G,\bm{\alpha},S} \subseteq S$
    of the notion of common knowledge of rationality for $G$, $\bm{\alpha}$ and $S$}
    as 
    \begin{align*}
        T_{G,\bm{\alpha},S} = \{ \bm{s}\in S \mid &\exists M=(W,(R_p)_{p\in P},\bm{\sigma})\in M(G,S). \\
        &\exists w \in W.\ (w\in \calCK \RAT_{G,\bm{\alpha},M,S} \wedge \bm{\sigma}(w) = \bm{s}) \}.
    \end{align*}
\end{definition}

\begin{lemma}\label{lem:winning}
    Let $T_{G,\bm{\alpha},S}$ be a characterization for
    a game arena $G=\Arena$, an objective profile $\bm{\alpha}$ and a class $S$ of strategy profiles.
    If $p\in P$ has a winning strategy,
    then $p$ is a winner, namely $p \in \Win(\bm{\alpha},\bm{t})$
    for all $\bm{t}\in T_{G,\bm{\alpha},S}$.
\end{lemma}

\begin{proof}
    Assume $p \in P$ has a winning strategy.
    Let $\bm{t}$ be an arbitrary strategy profile
    in $T_{G,\balpha,S}$.
    By the definition of $T_{G,\balpha,S}$,
    there exist
    $M=(W,(R_p)_{p\in P},\bsigma)\in M(G,S)$ and
    $w\in W$ such that
    $w\in \calCK \RAT_{G,\balpha,M,S}$ and
    $\bm{t} = \bsigma(w)$.
    Since $w\in R^+(w)$,
    $w\in \calCK \RAT_{G,\balpha,M,S}$ implies
    $w\in \RAT_{G,\balpha,M,S}$.
    If $p$ is a loser under~$\bm{t}$, namely $p \notin \Win(\balpha,\bsigma(w))$,
    then $p$ is not rational in~$w$
    because her winning strategy satisfies
    both conditions (\ref{align:rat1}) and (\ref{align:rat2}) in Definition~\ref{def:rat}
    by letting $w'=w$ in condition~(\ref{align:rat2}).
    This contradicts
    $w\in \RAT_{G,\balpha,M,S}$,
    and thus, $p$ should win under~$\bm{t}$.
\qed
\end{proof}
	\paragraph{}
		\begin{example}[continued]\label{exm:cr}
    Let $G$ and $M$ be the game arena and the epistemic model respectively, defined in Example~\ref{exm:model}.
    Let $\bm{\alpha'} = (O_p)_{p\in P}$ be an objective profile where
    $O_p = \Buchi(\{ v_p \})$.
    Then, it holds that $\RAT_{G,\bm{\alpha'},M,\Pos}=W$
    because for any world $w\in W$, player $p$ and $p$'s positional strategy $s_p\in \Pos_p$,
    the condition~(\ref{align:rat1}) in Definition~\ref{def:rat} holds but 
    (\ref{align:rat2}) does not hold.
    For example,
    let $w=RRR$, $p=0$ and $s_p = s^L_0$.
    Note that by the structure of $G$ and $\bm{\alpha'}$, player $0$ loses
    if and only if players $1$ and $2$ takes $s_1^L$ and $s_2^R$
    (regardless of the strategy of player $0$).
    Hence, $0\in\Win(\bm{\alpha'},\bm{\sigma}(w')) \iff 0\in\Win(\bm{\alpha'},\bm{\sigma}(w')[0\mapsto s^L_0])$
    for all $w' \in R_0(RRR)$.
    By $\RAT_{G,\bm{\alpha'},M,\Pos}=W$,
    it is easy to see that $\calCK\RAT_{G,\bm{\alpha'},M,\Pos}=W$ from the structure of $M$,
    and hence $T_{G,\bm{\alpha'},\Pos}=\{ \bm{\sigma}(w) \mid w \in W\} = \Pos$.
\end{example}
	\paragraph{Restriction of epistemic models}
		So far, we have made no assumptions about the size of an epistemic model and 
hence there could be an epistemic model whose size is extremely large. 
An epistemic model represents a structure of information that players have.
It is unnatural to assume that players can use extremely large information within a limited time or a limited computation power.
Therefore, we assume that there is a polynomial $p(n)$ such that
the size of a given epistemic model is not greater than $p(n)$ where $n$ is the size of a given game arena.

Let $G$ be a game arena and $S \in \{ \Pos,\Strategies \}$ be a class of strategy profiles.
We write the set of all epistemic models for $G$ and $S$
whose size is not greater than $p(n)$ for some polynomial $p$
as $M_{\classP}(G,S)$ where $n$ is the size of $G$.
Let $\bm{\alpha}$ be an objective profile.
We also define a characterization $T_{\classP,G,\bm{\alpha},S} \subseteq S$
as
\begin{align*}
    T_{\classP,G,\bm{\alpha},S} = \{ \bm{s}\in S \mid \exists &M=(W,(R_p)_{p\in P},\bm{\sigma})\in M_{\classP}(G,S). \\
    &\exists w \in W.\ (w\in \calCK \RAT_{G,\bm{\alpha},M,S} \wedge \bm{\sigma}(w) = \bm{s}) \}.
\end{align*}
\section{Verification Problems with Common Knowledge of Rationality}\label{sec:VPCKR}
	We define three types of rational verification problems. 
The first two of them are defined based on epistemic rationality while 
the last one is defined based on Nash equilibrium (NE). 
The second problem is a variant of the first problem,
where the size of an epistemic model is not greater than a polynomial size of a game arena.
We start with the analysis of the last problem
because NE is easier to analyze than epistemic rationality. 
	\begin{problem}\label{prob:VPCKR}
    We define \emph{verification problems with common knowledge of rationality (VPCKR)} as
    \begin{align*}
        \VPCKR_S &= \{ \langle G,\bm{\alpha},O \rangle \mid \forall \bm{t} \in T_{G,\bm{\alpha},S}.\ \out(\bm{t}) \in O \},\\
        \VPCKR_{\classP,S} &= \{ \langle G,\bm{\alpha},O \rangle \mid \forall \bm{t} \in T_{\classP,G,\bm{\alpha},S}.\ \out(\bm{t}) \in O \},\text{ and} \\
        \VPNash_S &= \{ \Tpl{G,\balpha,O} \mid
            \forall \bolds\in S.\, \Nash(\bolds,\balpha,S)
            \implies \Out(\bolds)\in O \}
    \end{align*}
    where $G=\Arena$ is a game arena,
    $\bm{\alpha}$ is an objective profile over Muller objectives,
    $O\subseteq \Play$ is a specification given by a Muller objective and
    $S \in \{ \Pos,\Strategies \}$ is a class of strategy profiles.
\end{problem}

\begin{example}[continued]
    Let $G$ and $\bm{\alpha'}=(O_p)_{p\in P}$ be the game arena and the objective profile, respectively
    defined in Example~\ref{exm:cr}.
    Let $O$ be the specification defined as $O = \bigcap_{p\in P}O_p$.
    Recall that $T_{G,\bm{\alpha'},\Pos}=\Pos$.
    Then, $\langle G,\bm{\alpha'},O \rangle \notin \VPCKR_\Pos$
    because there is the strategy profile $\bm{s_3}=(s^R_0,s^R_1,s^L_2)\in T_{G,\bm{\alpha'},\Pos}$
    such that $\out(\bm{s_3}) \notin O$.
\end{example}

Note that
$\VPCKR_{\classP,S}$ is \emph{not} a restricted problem of
$\VPCKR_{S}$, because
the fact that $\Tpl{G,\balpha,O} \notin \VPCKR_{S}$
(i.e.\ there is some $\bm{t}\in T_{G,\balpha,S}$
that satisfies $\Out(\bm{t})\notin O$)
gives no information on
whether every $\bm{t}\in T_{\classP,G,\balpha,S}$
satisfies $\Out(\bm{t})\in O$ or not.
For the same reason,
$\VPNash_S$ is not a restricted problem of
$\VPCKR_{S}$ or
$\VPCKR_{\classP,S}$.
On the other hand, we can say that
$\VPCKR_{S} \subseteq \VPCKR_{\classP,S} \subseteq \VPNash_S$ holds,
because every NE belongs to $T_{\classP,G,\bm{\alpha},S}$%
\footnote{
  For every NE $\bolds\in S$, the epistemic model with a single world $w$
  where $\bsigma(w)=\bolds$ satisfies
  $w \in \calCK \RAT_{G,\balpha,M,S}$.
}
and $T_{\classP,G,\balpha,S} \subseteq T_{G,\balpha,S}$.
Also note that
$\VPCKR_{\Pos}$, $\VPCKR_{\classP,\Pos}$, and
$\VPNash_{\Pos}$ are not restricted problems of
$\VPCKR_{\Strategies}$,
$\VPCKR_{\classP,\Strategies}$, and
$\VPNash_{\Strategies}$, respectively.
Moreover, because
$\bm{t}\in T_{G,\balpha,\Pos}$
does not imply
$\bm{t}\in T_{G,\balpha,\Strategies}$,
$\VPCKR_{\Strategies} \not\subseteq \VPCKR_{\Pos}$
in general.
(For the same reason,
$\VPCKR_{\classP,\Strategies} \not\subseteq \VPCKR_{\classP,\Pos}$
and $\VPNash_{\Strategies} \not\subseteq \VPNash_{\Pos}$
in general.)


Before investigating the complexity of
$\VPCKR_{\classP,S}$ and $\VPCKR_{S}$,
we mention the complexity of $\VPNash_S$.
As described in Introduction,
$\VPNash_S$ is closely related to rational synthesis (RS),
and with the class of Muller objectives
(which is closed under negation and the negation does not cause
exponential blow-up),
we can easily show that
the complexity of RS with the class $S$ of strategy profiles
is the same as that of $\OL{\VPNash_S}$.
By the results of~\cite{CFGR16},
we have the following proposition
for~$\VPNash_{\Strategies}$.
\begin{proposition}
$\VPNash_{\Strategies}$ is PSPACE-complete.
\end{proposition}

Although the complexity of the same problem
with positional strategies
was not studied in~\cite{CFGR16},
we can show that $\VPNash_{\Pos}$ is
$\classPi_2$-complete as follows.
	  \begin{lemma}\label{lem:VPNash_Pos_hard}
    $\VPNash_{\Pos}$ is $\classPi_2$-hard.
  \end{lemma}
  \begin{proof}
    We reduce $\aeSAT$ to $\VPNash_{\Pos}$.
    Let $\varphi =
      \forall x_1\ldots x_n \exists y_1\ldots y_m\,\psi$
    be an instance of $\aeSAT$,
    where $x_1,\ldots,x_n$, $y_1,\ldots,y_m$ are variables
    and $\psi$ is a Boolean formula over them.
    From $\varphi$, we construct an instance
    $\Tpl{G,\balpha,O}$ of $\VPNash_{\Pos}$
    where $G=(\{A,E\}, V_A\cup V_E, (V_A,V_E), a_1, \Delta)$
    as follows.
    \begin{align*}
      V_A &= \{a_1,\ldots,a_n,\, x_1,\ldots,x_n,\,
               \OL{x_1},\ldots,\OL{x_n}\}, \\
      V_E &= \{e_1,\ldots,e_m,\, y_1,\ldots,x_m,\,
               \OL{y_1},\ldots,\OL{y_m}\}, \\
      \Delta &=
        \{(a_i,u) \mid 1\le i\le n,\ u\in\{x_i,\OL{x_i}\}\} \\
        &\:{}\cup
        \{(u,a_{i+1}) \mid 1\le i < n,\ u\in\{x_i,\OL{x_i}\}\}
        \cup \{(u,e_1) \mid u\in\{x_n,\OL{x_n}\}\} \\
        &\:{}\cup
        \{(e_i,u) \mid 1\le i \le n,\ u\in\{y_i,\OL{y_i}\}\} \\
        &\:{}\cup
        \{(u,e_{i+1}) \mid 1\le i < n,\ u\in\{y_i,\OL{y_i}\}\}
        \cup \{(u,a_1) \mid u\in\{y_m,\OL{y_m}\}\}, \\
      \balpha &= (O_{\!A},O_E)
        \text{ where }O_{\!A}=\True\text{ and }O_E=\psi\text{, and} \\
      O &= \psi.
    \end{align*}
    Figure~\ref{fig:pos_hard} shows the game arena obtained from a formula over
    $x_1$, $x_2$ and $y_1$.
    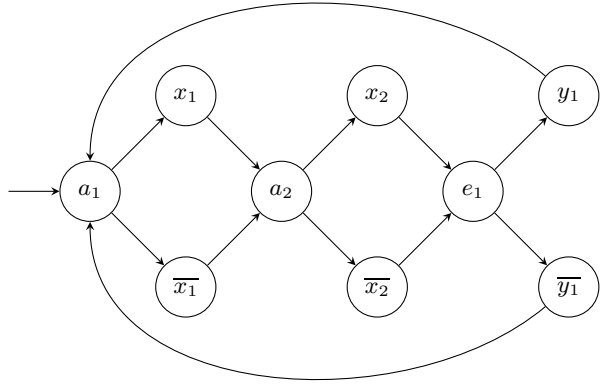
\begin{figure}[t]
		\centering
		\begin{tikzpicture}[everynode/.style={circle,draw,minimum size=0.8cm},>=stealth,node distance=1.8cm]
            \node[everynode] (v1) {$a_1$};
			\node[left of=v1,node distance=1.2cm] (start) {};
            \node[everynode, above right of=v1] (x1) {$x_1$};
            \node[everynode, below right of=v1] (nx1) {$\overline{x_1}$};
            \node[everynode, below right of=x1] (v2) {$a_2$};
            \node[everynode, above right of=v2] (x2) {$x_2$};
            \node[everynode, below right of=v2] (nx2) {$\overline{x_2}$};
            \node[everynode, below right of=x2] (v3) {$e_1$};
            \node[everynode, above right of=v3] (x3) {$y_1$};
            \node[everynode, below right of=v3] (nx3) {$\overline{y_1}$};
			\draw[->] (start) to (v1);
			\draw[->] (v1) to (x1);
			\draw[->] (v1) to (nx1);
			\draw[->] (x1) to (v2);
			\draw[->] (nx1) to (v2);
			\draw[->] (v2) to (x2);
			\draw[->] (v2) to (nx2);
			\draw[->] (x2) to (v3);
			\draw[->] (nx2) to (v3);
			\draw[->] (v3) to (x3);
			\draw[->] (v3) to (nx3);
			\draw[->] (x3) to [out=140, in=90] (v1);
			\draw[->] (nx3) to [out=220, in=270] (v1);
		\end{tikzpicture}
			\caption{The game arena constructed from a formula having $x_1$, $x_2$ and $y_1$}
			\label{fig:pos_hard}
	\end{figure}
    Note that we regard a Boolean formula $\psi$
    as a Muller objective.
    For example, a formula $\psi = x_1\lor\OL{x_2}$ is considered as
    the Muller objective such that
    a player whose objective is $\psi$ wins if
    the play visits vertex $x_1$ infinitely often or
    the play visits $x_2$ only finite times.
  
    By the structure of $G$,
    we can consider
    every strategy profile $\bolds\in\Pos$ over positional strategies
    as a truth assignment to the variables in $\varphi$\,;
    choosing $x_i$ (resp.\ $\OL{x_i}$) as the next vertex at vertex $a_i$
    corresponds to letting $x_i = \True$ (resp.\ $x_i = \False$).
    The same for $y_i$ and $\OL{y_i}$.
  %
  %
    The play $\Out(\bolds)$ contains
    the chosen vertices infinitely many times while
    it does not contain any unchosen vertex.
    Therefore, $\Out(\bolds)$ satisfies $\psi$ as a Muller objective
    if and only if $\psi$ is true under the truth assignment
    represented by~$\bolds$.
  
    We show that $\varphi \in \aeSAT \iff
    \Tpl{G,\balpha,O} \in \VPNash_{\Pos}$.
  
    ($\Longrightarrow$) \
    Assume that $\varphi\in\aeSAT$.
    We have to show
    $\forall \bolds \in\Pos.\,\Nash(\bolds,\balpha,\Pos)$ $\implies
    \Out(\bolds)\in O$.
    Assume that $\bolds = (s_{\!A},s_E) \in \Pos$ and
    $\Nash(\bolds,\balpha,\Pos)$.
    Since $\varphi\in\aeSAT$,
    for the truth assignment represented by $s_{\!A}$,
    there must exist a truth assignment to
    $y_1,\ldots,y_m$ that makes $\psi$ true.
    Let $s'_E$ denote
    the positional strategy of $E$ corresponding to
    this assignment;
    hence, $\Out(\Deriv{\bolds}{E}{s'_E})\in O_E$.
    On the other hand,
    by the definition of NE, either
    $\Out(\bolds)\in O_E$ or
    $\Out(\Deriv{\bolds}{E}{s'_E})\notin O_E$ should hold
    for any positional strategy $s'_E$ of~$E$.
    As shown above, the latter does not hold.
    Therefore, the former holds and thus
    $\Out(\bolds)\in O$ since $O = O_E$.
  
    ($\Longleftarrow$) \
    Assume that
    $\forall \bolds \in\Pos.\,\Nash(\bolds,\balpha,\Pos) \implies
    \Out(\bolds)\in O$.
    Let $s_{\!A}$ be an arbitrary positional strategy of~$A$.
    We have to show that
    for $s_{\!A}$, there exists a positional strategy $s_E$ of
    $E$ such that
    the truth assignment corresponding to $(s_{\!A},s_E)$
    makes $\psi$ true.
    Let $s_E$ be an arbitrary chosen positional strategy of~$E$
    and let $\bolds = (s_{\!A},s_E)$.
    By the assumption,
    $\Out(\bolds)\in O$ holds or $\Nash(\bolds,\balpha,\Pos)$ does not hold.
    If $\Out(\bolds)\in O$ holds,
    then $s_E$ is just the desired strategy that
    makes $\psi$ true.
    If $\Nash(\bolds,\balpha,\Pos)$ does not hold,
    then there exist a player $p$ and
    a positional strategy $s'$ of $p$ that satisfy
    $\Out(\bolds)\notin O_p$ and $\Out(\Deriv{\bolds}{p}{s'})\in O_p$.
    We have $p=E$ because $\Out(\bolds)\notin O_{\!A}$ never holds.
    Since $\Out(\Deriv{\bolds}{E}{s'})\in O_E$,
    $s'$ is the desired strategy that makes $\psi$ true.
  \qed
  \end{proof}
	\begin{lemma}\label{lem:VPNash_Pos}
    $\VPNash_{\Pos}$ is in $\classPi_2$.
  \end{lemma}
  \begin{proof}
    By Lemma~\ref{lem:sver_pos}, deciding whether $\Out(\bolds)\in O$
    is in~$\classP$
    for a given game, a Muller objective $O$,
    and a strategy profile $\bolds$ over
    positional strategies.
    Deciding whether $\Nash(\bolds,\balpha,\Pos)$ for a given game
    and $\bolds$ is in $\coNP$\@:
    Guess a player $p$ and a positional strategy $s_p$ of $p$,
    and check whether $\Out(\bolds)\notin O_p$ and
    $\Out(\Deriv{\bolds}{p}{s_p})\in O_p$.
    Let $A$ be an oracle of this problem.
    Using $A$, we can construct a non-deterministic
    polynomial-time oracle Turing machine for deciding
    whether $\Tpl{G,\balpha,O}\notin\VPNash_{\Pos}$:
    Guess $\bolds$ and check whether
    $\Nash(\bolds,\balpha,\Pos)$ and
    $\Out(\bolds)\notin O$.
    Therefore, $\VPNash_{\Pos}$ is in
    $\coNP^{\NP} = \classPi_2$.
  \qed
  \end{proof}

  \begin{theorem}
  $\VPNash_{\Pos}$ is $\classPi_2$-complete.
\end{theorem}
\begin{proof}
  By Lemmas~\ref{lem:VPNash_Pos_hard} and~\ref{lem:VPNash_Pos}.
\qed
\end{proof}

Next, let us consider the complexity of
$\VPCKR_{S}$ and $\VPCKR_{\classP,S}$.
In this paper, we concentrate on the complexity of
$\VPCKR_{\Pos}$ and $\VPCKR_{\classP,\Pos}$.
For $\VPCKR_{\classP,\Pos}$,
we can show that it is coNP-hard and in $\classPi_2$
as follows.

\begin{theorem}\label{th:VPCKR_P_hard}
    $\VPCKR_{\classP,\Pos}$ is coNP-hard.
\end{theorem}

\begin{proof}
    We reduce $\overline{\SAT}$, the complement of the satisfiability problem,
    to $\VPCKR_{\classP,\Pos}$.
    Let $\varphi$ be a Boolean formula given as an instance of $\OL{\SAT}$, where
    $x_1,\dots x_n$ are variables of $\varphi$.
    From $\varphi$, we construct an instance
    $\Tpl{G,\balpha,O}$ of $\VPCKR_{\classP,\Pos}$
    where $G=(\{p\},V,(V),v_1,\Delta)$ as follows.
    \begin{align*}
         V &= \{ v_1,\dots,v_n,\, x_1,\dots,x_n,\, \OL{x_1},\dots,\OL{x_n} \}, \\
         \Delta &= 
			\{ (v_i,u) \mid 1\leq i \leq n,\, u\in \{x_i,\OL{x_i}\}\} \\
            &\:{}\cup
			\{ (u,v_{i+1}) \mid 1\leq i < n,\, u\in \{x_i,\OL{x_i}\}\} \\
            &\:{}\cup
			\{ (u,v_1) \mid u\in \{x_i,\OL{x_i}\}\}, \\
         \balpha &=(O_p)\text{ where }O_p = \varphi,\, \text{ and }\\
		 O & =\neg\varphi.
    \end{align*}
    We regard a Boolean formula $\varphi$ as a Muller objective
    and consider every strategy profile $\bm{s}\in\Pos$ over positional strategies
    as a truth assignment to the variables in $\varphi$ in the same way as
    the proof of Lemma~\ref{lem:VPNash_Pos_hard}.

    We show that $\varphi \in \OL{\SAT} \iff \Tpl{G,\balpha,O} \in \VPCKR_{\classP,\Pos}$.

    ($\Longrightarrow$) \
    Assume $\varphi \in \OL{\SAT}$.
    Because 
    $\varphi$ is unsatisfiable,
    any strategy profile $\bm{s} \in \Pos$ satisfies $\out(\bm{s}) \in O$ ($=\neg \varphi$).
    Therefore, $\forall \bm{t} \in T_{\classP,G,\bm{\alpha},\Pos}.\ \out(\bm{t}) \in O$ holds.

    ($\Longleftarrow$) \
    Assume $\forall \bm{t}\in T_{\classP,G,\balpha,\Pos}.\ \out(\bm{t})\in O$.
    Note that $\varphi \in \OL{\SAT}$ is equivalent to $\forall \bm{s} \in \Pos.\ \out(\bm{s}) \in O$.
    We show this by contradiction.
    Suppose $\bm{s} \in \Pos$ and $\out(\bm{s}) \notin O$.
    Consider an epistemic model
    $M=(\{w\},(R_p)_{p\in P},\bm{\sigma})$
    where $R_p= \{ (w,w) \}$ and $\bm{\sigma}(w) = \bm{s}$.
    Because $\out(\bm{s})\notin O$, player $p$ wins in world~$w$.
    It is easy to see that $w \in \calCK \RAT$ because of the structure of $M$.
    Therefore, $\bm{s}\in T_{\classP,G,\bm{\alpha},\Pos}$.
    By the outer assumption, $\out(\bm{s})\in O$,
    which contradicts the inner assumption $\out(\bm{s}) \notin O$.
\qed
\end{proof}
	\begin{theorem}\label{th:VPCKR_P}
    $\VPCKR_{\classP,\Pos}$ is in $\classPi_2$.
\end{theorem}

\begin{proof}
    In a similar way to the proof of Lemma~\ref{lem:VPNash_Pos},
    we can construct a non-deterministic
    polynomial-time oracle Turing machine for deciding
    $\OL{\VPCKR_{\classP,\Pos}}$.
  %
  %
    Deciding whether $w\in\calCK\RAT_{G,\bm{\alpha},M,\Pos}$ for a given game,
    an epistemic model $M=(W,(R_p)_{p\in P},\bm{\sigma})$ and a world $w\in W$ is in $\coNP$\@:
    Guess a world $u\in R^+(w)$, a player $p$ and a positional strategy $s_p\in\Pos_p$ of $p$,
    and check 
    both conditions (\ref{align:rat1}) and (\ref{align:rat2}) in Definition~\ref{def:rat}
    substituting $u$ for $w$.
  %
    A non-deterministic
    polynomial-time oracle Turing machine for deciding
    whether $\Tpl{G,\balpha,O}\notin\VPCKR_{\classP,\Pos}$
    is as follows:
    Guess an epistemic model $M=(W,(R_p)_{p\in P},\bm{\sigma}) \in M_{\classP}(G,S)$ and
    a world $w\in W$, and check whether
    $w\in \calCK \RAT_{G,\bm{\alpha},M,\Pos}$ and
    $\out(\bm{\sigma}(w)) \notin O$.
    Note that because the size of $M$ is not greater than some polynomial of the size of $G$,
    the construction of $M$ takes only polynomial time.
    Therefore, $\VPCKR_{\classP,\Pos}$ is in
    $\coNP^{\NP} = \classPi_2$.
  \qed
\end{proof}


For $\VPCKR_{\Pos}$,
we can show that it is $\classSigma_2$-hard
and in $\coNEXP^{\NP}$
as follows.

\begin{theorem}
  $\VPCKR_{\Pos}$ is $\classSigma_2$-hard.
\end{theorem}
\begin{proof}
  We reduce $\eaSAT$ to $\VPCKR_{\Pos}$.
  Let $\varphi =
    \exists y_1\dots y_m\, \forall x_1\dots x_n\,\psi$
  be an instance of $\eaSAT$.
  From $\varphi$, we construct an instance
  $\Tpl{G,\balpha,O}$ of $\VPCKR_{\Pos}$
  where
  $G$ is the same game arena as in
  the proof of Lemma~\ref{lem:VPNash_Pos_hard}
  and $O=\psi$ and
  $\balpha$ consists of
  $O_A=\neg\psi$ and $O_E=\psi$.
  As described in the proof of Lemma~\ref{lem:VPNash_Pos_hard},
  there is a one-to-one correspondence between
  the strategy profiles in $\Pos$ and
  the truth assignments to the variables in~$\varphi$,
  and player $E$ wins under a strategy profile $\bolds$
  if and only if
  $\psi$ is true under the truth assignment corresponding to~$\bolds$.
  Therefore
  by the structure of~$\varphi$,
  $\varphi$ is true if and only if
  $E$ has a winning strategy.

  We show that $\varphi \in \eaSAT \iff
  \Tpl{G,\balpha,O} \in \VPCKR_{\Pos}$.

  ($\Longrightarrow$) \
  As mentioned above,
  $\varphi \in \eaSAT$ implies
  $E$ has a winning strategy.
  By Lemma~\ref{lem:winning},
  $E$ wins for every $\bm{t}\in T_{G,\balpha,\Pos}$\,;
  i.e.,
  $\forall \bm{t}\in T_{G,\balpha,\Pos}.\,
   \Out(\bm{t})\in O_E$ ($=O$).

  ($\Longleftarrow$) \
  Assume that
  $\varphi\notin\eaSAT$;
  that is, $E$ has no winning strategy.
  If $A$ has a winning strategy $s_{\!A}$,
  then any strategy profile $\bolds\in\Pos$
  where $A$ takes $s_{\!A}$ is
  an NE\@.
  Since every NE belongs to $T_{G,\balpha,\Pos}$
  and $A$ wins under $\bolds$,
  it holds that
  $\bolds\in T_{G,\balpha,\Pos}$ and
  $\Out(\bm{t})\notin O$.
  Therefore, $\Tpl{G,\balpha,O}\notin \VPCKR_{\Pos}$.

  Consider the case where
  neither $A$ nor $B$ has a winning strategy.
  Let $T^{\infty}$ be the subset of strategy profiles
  obtained by
  the following iterative procedure
  called the \emph{iterated deletion of inferior
  profiles} (IDIP)~\cite[Def.~9.8]{Bo15}:
  For a subset $X$ of strategy profiles and
  its member $\bolds=(s_p)_{p\in P}\in X$,
  $\bolds$ is \emph{inferior relative to} $X$ if
  there exist a player $p$ and $p$'s strategy $t_p\in\Pos_p$
  such that
  \begin{enumerate}
  \item $\Out(\bolds)\notin O_p \land
         \Out(\Deriv{\bolds}{p}{t_p})\in O_p$, and
  \item $\Out(\bolds')\in O_p \implies
         \Out(\Deriv{\bolds'}{p}{t_p})\in O_p$
        for every $\bolds'=(s'_p)_{p\in P}\in X$ such that
        $s'_p = s_p$.
  \end{enumerate}
  Let $T^0=\Pos$.
  $T^{i+1}$ is the set obtained from $T^i$ by removing
  all strategy profiles inferior relative to~$T^i$.
  We repeat this construction until
  there is no inferior strategy profiles.
%
  Since $\Pos$ is finite,
  this procedure always halts.
  Moreover,
  since neither $A$ nor $B$ has a winning strategy
  and the game is zero-sum,
  there must exist a strategy profile $\bolds\in\Pos$
  remaining in $T^{\infty}$ that satisfies
  $\Out(\bolds)\notin O_E$ ($=O$).
  As shown in \cite[Proposition~9.4~(B)]{Bo15},
  $\bolds\in T^{\infty}$ also belongs to
  $T_{G,\balpha,\Pos}$.
  Therefore, $\Tpl{G,\balpha,O}\notin \VPCKR_{\Pos}$.
\qed
\end{proof}

\begin{theorem}
    $\VPCKR_{\Pos}$ is in $\coNEXP^{\NP}$.
\end{theorem}

\begin{proof}
    For a subset $X$ of $\Pos$,
    let $M(X)=(W,(R_p)_{p\in P},\bsigma)$
    be an epistemic model
    where $W=X$ and
    $\bsigma(\bolds)=\bolds$ for all $\bolds\in X$,
    and
    $R_p = \{(w_1,w_2)\mid \sigma_p(w_1)=\sigma_p(w_2)\}$ for
    $p\in P$.
    Using this construction of an epistemic model
    from a subset of strategy profiles,
    we can construct a non-deterministic
    exponential-time
    oracle Turing machine for deciding
    $\OL{\VPCKR_{\Pos}}$ as follows:
    Let $\Tpl{G,\balpha,O}$ be a given instance of $\VPCKR_{\Pos}$.
    Guess a subset $T\subseteq\Pos$ and $\boldt\in T$,
    and construct $M(T)$.
    Then, check whether $\boldt$ (as a world of $M(T)$)
    satisfies $\boldt\in\calCK\RAT_{G,\balpha,M(T),\Pos}$
    and $\Out(\boldt)\notin O$.

    As mentioned in the proof of Theorem~\ref{th:VPCKR_P},
    deciding whether $\boldt\in\calCK\RAT$ for given
    $M(T)$ and $\boldt$ is in $\coNP$.
    The above Turing machine uses an $\NP$ oracle to
    decide $\boldt\in\calCK\RAT$.
    Since the size of $T$ and $M(T)$ is exponential to the size of
    $G$ in general,
    guessing $T$ and constructing $M(T)$ take exponential time.
    Deciding whether $\Out(\boldt)\notin O$ is in $\classP$
    by Lemma~\ref{lem:sver_pos} and is not dominant.

    If the answer of the Turing machine is \emph{yes},
    then obviously $\boldt\in T_{G,\balpha,\Pos}$ and thus
    $\Tpl{G,\balpha,O}\in \OL{\VPCKR_{\Pos}}$.
    On the other hand,
    \cite[Proposition~9.4]{Bo15} and its proof say that
    the subset $T^{\infty}$ of strategy profiles
    obtained by the IDIP procedure~\cite[Def.~9.8]{Bo15}
    satisfies $T^{\infty}=T_{G,\balpha,\Pos}$, and
    every world $w$ of the epistemic model $M(T^{\infty})$
    satisfies $w\in\calCK \RAT_{G,\balpha,M(T^{\infty}),\Pos}$.
    Hence, we do not need to consider epistemic models
    other than $M(T)$ for $T\subseteq\Pos$.
    Therefore,
    if the answer of the above Turing machine is \emph{no},
    then we can conclude that
    $\Tpl{G,\balpha,O}\in \VPCKR_{\Pos}$.
  \qed
\end{proof}

\section{Conclusion}\label{sec:conclusion}
	We introduced an epistemic approach to rational verification.
We defined rational verification problems
$\VPCKR_S$, $\VPCKR_{\classP,S}$ and $\VPNash_S$ based on
common knowledge of rationality and Nash equilibrium.
The problem $\VPCKR_{\classP,S}$ is a variant of $\VPCKR_S$
where the size of an epistemic model is not greater than $p(n)$ for some polynomial $p$ and
the size $n$ of a given game arena.
The problem $\VPNash_S$ asks whether
each Nash equilibrium satisfies given specification.
Then, we analyzed the complexities of these problems shown.
Table~\ref{tab:complexity} summarizes the complexities of the problems.

In this paper, we consider only for the $S=\Pos$ case.
Analysing above problems for other $S$ such as the class of the finite memory strategies
is future work.
Our epistemic model based on KT5 Kripke frame is the knowledge based setting.
Hence, if a player knows $E$ in a world $w$, then $E$ actually occurs in $w$.
There is another setting called belief based setting.
In belief based setting, even if a player knows $E$ in $w$, $E$ doesn't necessarily occur in $w$.
Studying belief based setting is also future work.


\begin{thebibliography}{99}

\bibitem{BCJ18} 
R. Bloem, K. Chatterjee and B. Jobstmann, 
Graph Games and Reactive Synthesis, 
E. M. Clarke et al. (eds.), Handbook of Model Checking, Chapter 27, 921--962, Springer, 2018. 


\bibitem{CGP01}
E. M. Clarke, O. Grumberg and D. A. Peled, Model Checking. MIT Press, 2001.

\bibitem{Bo15}
G. Bonanno,
Epistemic Foundations of Game Theory,
H. van Ditmarsch et al.\ (eds.), 
Handbook of Epistemic Logic,
Chapter 9, 411--450, College Publications, 2015.

\bibitem{BL69} 
J. R. B\"{u}chi and L. H. Landweber, 
Solving sequential conditions by finite-state strategies, 
Trans. American Mathematical Society 138, 295--311, 1969. 

\bibitem{PR89} 
A. Pnueli and R. Rosner, 
On the synthesis of a reactive module, 
16th ACM Symp. on Principles of Programming Languages (POPL 1989), 179--190. 


\bibitem{FKL10} 
D. Fisman, O. Kupferman and O. Lustig, 
Rational synthesis, 
16th Int. Conf. on Tools and Algorithms for the Construction and Analysis of Systems (TACAS 2010), 
LNCS 6015, 190--204.


\bibitem{GNPW23} 
J. Gutierrez, M. Najib, G. Perelli and M. Wooldridge, 
On the complexity of rational verification, 
Annals of Mathematics and Artificial Intelligence 91, 409--430, 2023. 

\bibitem{Um08} 
M. Ummels, 
The complexity of Nash equilibria in infinite multiplayer games, 
11th Int. Conf. on Foundations of Software Science and Computational Structures (FOSSACS 2008), 
LNCS 4962, 20--34. 


\bibitem{CFGR16} 
R. Condurache, E. Filiot, R. Gentilini and J-F. Raskin, 
The complexity of rational synthesis, 
43rd Int. Colloq. on Automata, Languages, and Programming (ICALP 2016), 
LIPIcs 55, 121:1--121:15. 


\bibitem{KR01}
S. Kremer and J.-F. Raskin, 
A game-based verification of non-repudiation and fair exchange protocols, 
12th Int. Conf. on Concurrency Theory (CONCUR 2001), LNCS 2154, 551--565. 
extended version: J. Computer Security 11(3), 399--429, 2003. 

\bibitem{CR12} 
K. Chatterjee and V. Raman, 
Synthesizing protocols for digital contract signing, 
13th Int. Conf. on Verification, Model checking, and Abstract Interpretation (VMCAI 2012), 
LNCS 7148, 152--168. 


\bibitem{KPV16} 
O. Kupferman, G. Perelli and M. Vardi, 
Synthesis with rational environments, 
Annals of Mathematics and Artificial Intelligence 78(1), 3--20, 2016. 

\bibitem{KS22} 
O. Kupferman and N. Shenwald,
The complexity of LTL rational synthesis,
28th Int. Conf. on Tools and Algorithms for the Construction and Analysis of Systems (TACAS 2022), 
LNCS 13243, 25--45.

\bibitem{BRT22} 
V. Bruy\`{e}re, J.-F. Raskin and C. Tamines, 
Pareto-rational verification, 
33rd Int. Conf. on Concurrency Theory (CONCUR 2022), 
LIPIcs.CONCUR.2022, 33:1--33:20. 

\bibitem{BRB23}
L. Brice, J.-F. Raskin and M. van den Bogaard, 
Rational verification for Nash and subgame-perfect equilibria in graph games, 
48th Int. Symp. on Mathematical Foundations of Computer Science (MFCS 2023), 
26:1--26:15. 


\bibitem{AB95}
R. Aumann and A. Brandenburger,
Epistemic conditions for Nash equilibrium,
Econometrica: Journal of the Econometric Society,
1995,
1161--1180.

\bibitem{P99}
B. Polak,
Epistemic conditions for Nash equilibrium, and common knowledge of rationality,
Econonetrica,
67.3,
1999,
673--676.
\end{thebibliography}
\end{document}